\documentclass[twocolumn,conference]{IEEEtran}

\usepackage{latexsym}
\usepackage{bm}
\usepackage[dvips]{graphics}
\usepackage{graphicx,subfigure}
\usepackage{psfrag}
\usepackage{amsfonts,amsmath,amssymb}
\usepackage{algorithm,algorithmic}
\usepackage{dsfont,color,subfigure}
 \usepackage{pslatex}

\usepackage{xspace}
\usepackage{bbm}

\def\In{\mathop{\rm In}\nolimits}%
\def\Out{\mathop{\rm Out}\nolimits}%

\newcommand{\Ac}{\mathcal{A}}

\newcommand{\Ec}{\mathcal{E}}
\newcommand{\Fc}{\mathcal{F}}

\newcommand{\Ic}{\mathcal{I}}

\newcommand{\Nc}{\mathcal{N}}

\newcommand{\Rc}{\mathcal{R}}
\newcommand{\Sc}{\mathcal{S}}
\newcommand{\Tc}{\mathcal{T}}

\newcommand{\Vc}{\mathcal{V}}
\newcommand{\Wc}{\mathcal{W}}
\newcommand{\Xc}{\mathcal{X}}
\newcommand{\Yc}{\mathcal{Y}}

\def\a{\alpha}
\def\b{\beta}

\def\d{\delta}
\def\e{\epsilon}

\let\P\relax
\DeclareMathOperator\P{P}

\newcommand\ie{i.e.,\xspace}
\def\textiid{i.i.d.\@\xspace}
\newcommand\iid{\ifmmode\text{ i.i.d. } \else \textiid \fi}

\newtheorem{theorem}{Theorem}
\newtheorem{lemma}{Lemma}

\begin{document}

\title{On Equivalence Between Network Topologies}
%

\author{\IEEEauthorblockN{Tracey Ho}
\IEEEauthorblockA{Department of Electrical Engineering\\
California Institute of Technology\\
tho@caltech.edu; } \and
\IEEEauthorblockN{Michelle Effros}
\IEEEauthorblockA{Departments of Electrical Engineering\\
California Institute of Technology\\
effros@caltech.edu;} \and
\IEEEauthorblockN{Shirin Jalali}
\IEEEauthorblockA{Center for Mathematics of Information\\
California Institute of Technology\\
shirin@caltech.edu;}}

\maketitle

\newcommand{\p}{\mathds{P}}
\newcommand{\mb}{\mathbf{m}}
\newcommand{\bb}{\mathbf{b}}

\begin{abstract}
One major open problem in network coding is  to characterize the capacity region of a general multi-source multi-demand network. There are some existing computational tools for bounding the capacity of general networks, but their computational complexity grows very quickly with the size of the network. This motivates us to propose a new hierarchical approach which finds upper and lower bounding networks of smaller size for a given network. This approach sequentially replaces components of the network with simpler structures, i.e., with fewer links or nodes, so that the resulting network is more amenable to computational analysis and its capacity provides an upper or lower bound on the capacity of the original network. The accuracy of the resulting bounds can be bounded as a function of the link capacities. Surprisingly, we are able to simplify some families of network structures without any loss in accuracy.
\end{abstract}

\section{Introduction}

Finding the network coding capacity of networks with general topologies and communication demands is a challenging open problem, even for
networks consisting of noiseless point-to-point links. Information theoretic inequalities can be used to bound network capacities, but it is in
general a complex task to find the best combination of inequalities to apply.  While various bounds (e.g.~\cite{KramerS:06,AdlerH:06,AdlerH:06,HarveyK:07}) can be obtained by clever choices of inequalities, we would like to have systematic general techniques  for bounding the capacity in arbitrary network problems. We hope to derive these bounds in a way that allows us to bound the accuracy of the obtained  bounds and to trade off tightness and computation. The LP outer bound~\cite{Yeung:97},
which gives the tightest outer bound implied by Shannon-type inequalities and has been implemented as the software programs Information Theoretic
Inequalities Prover (ITIP) \cite{SubramanianT:08} and XITIP \cite{xitip}, has complexity exponential in the number of links in the network and can thus only be used to compute capacity
bounds for relatively small problem instances.  Inner bounds can be obtained by restricting attention to scalar linear, vector linear or abelian codes e.g.~\cite{MedardK:03,Chan:07}, but the complexity of such approaches also grows quickly in the network size. This motivates us to seek systematic
tools for bounding the capacity of a network by the capacity of another network with fewer links and  characterizing the difference in capacity.

In this paper we introduce a novel approach for analyzing  capacity regions of acyclic networks consisting of capacitated noiseless links with general demands. Inspired by \cite{KoetterE:10}, we employ a  hierarchical network analysis that replaces pieces of the network by equivalent or bounding models with fewer links.  At each step of the process, one component of the network is replaced by a simpler structure with the same inputs and outputs. The model is an upper bounding model if all functions that can be implemented across the original network can also be implemented across the model. The model is a lower bounding model if all functions that can be implemented across that model can also be implemented across the given component. If the same model is both an upper bounding model and a lower bounding model for a given component, then the component and its model are equivalent.  Where possible, we try to find upper and lower bounds that have identical structures, since bounding the accuracy of the resulting capacity bounds is easier when the topologies of the upper and lower bounding networks match.

The organization of this paper is as follows. Section \ref{s:model} describes the system model. The problem of finding equivalent or bounding networks of smaller size and the properties of such networks is discussed in Section \ref{s:e-b-networks}. Sections \ref{s:simple}  and  \ref{s:Y} describe a variety of operations for finding such networks. Section \ref{s:effect}  treats accuracy bounds. The networks considered in this paper are assumed to be acyclic. The effect of cycles and delay is discussed in Section \ref{s:cycle}. Finally, Section \ref{s:conclude} concludes the paper.

\section{System model}\label{s:model}
We mainly use the model and notation introduced in \cite{book:Yeung}. The network is modeled by an acyclic directed graph $\Nc=(\Vc,\Ec)$, where $\Vc$ and $\Ec\subseteq\Vc\times\Vc$ denote the set of nodes and links, respectively. Each directed link $e=(v_1,v_2)\in\Ec$ represents a noiseless link of capacity $C_e$ between the nodes $v_1$ and $v_2$ in $\Nc$.  For each node $v$, let $\In(v)$ and $\Out(v)$ denote the set of incoming and outgoing links of node $v$ respectively.

We assume that the source nodes ($\Sc$) and sink nodes ($\Tc$) are distinct, i.e., $\Sc\bigcap\Tc=\emptyset$, and each source (sink) node has only outgoing (incoming) links. There is no loss of generality in this assumption since any network that violates these conditions can be modified to form a network that satisfies these conditions and has the same capacity. After these modifications, each node $v\in \Vc$ falls into one of the following categories: i) source nodes ($\Sc$), ii) sink nodes ($\Tc$) and iii) relay nodes ($\Ic$).  Relay nodes have both incoming and outgoing links, and they do not observe any external information. Their role is to facilitate data transfer from the source nodes to the sink nodes.

A subnetwork $\Nc_s=(\Vc_s,\Ec_s)$ of network $\Nc=(\Vc,\Ec)$ is  constructed based on a subset of relay nodes $\Ic_s\subset\Ic$ as follows. The set of sources and sinks  of the  subnetwork $\Nc_s$ are defined as $\Sc_s=\{v\in\Vc: v\notin \Ic_s, \exists v'\in\Ic_s, (v,v')\in\Ec\}$, and $\Tc_s=\{v\in\Vc: v\notin\Ic_s, \exists v'\in\Ic_s, (v',v)\in\Ec\}$. Then, $\Vc_S=\Sc_s\bigcup\Ic_s\bigcup\Tc_s$, and $\Ec_s=\{e\in\Ec: e=(v,v'), v,v'\in\Vc_s\}$.

A coding scheme of block length $n$ for this network is described as follows. Each source node $s\in\Sc$ observes some message $M_s\in\Xc_s=\{1,2,\ldots,2^{nR_s}\}$. Each sink node $t\in\Tc$ is interested in recovering some of the messages that are observed by source nodes. Let $\b(t)\subseteq\Sc$ denote the set of source nodes that the node $t$ is interested in recovering. The coding operations performed by each node can be categorized as follows
\begin{enumerate}
\item Encoding functions performed by the source nodes:  For each $s\in\Sc$, and $e\in\Out(s)$, the encoding function corresponding to link $e$ is described as
\begin{align}
g_e:\Xc_{s}\to\{1,2,\ldots,2^{nC_e}\}.
\end{align}
\item Relay functions performed at realy nodes: If $v\notin\Sc\bigcup\Tc$, then for each $e\in\Out(v)$, the relay function corresponding to the link $e$ is described as
\begin{align}
g_e:\prod\limits_{e'\in\In(v)}\{1,2,\ldots,2^{nC_{e'}}\} \to \{1,2,\ldots,2^{nC_e}\}.
\end{align}
\item Finally, for each $t\in\Tc$, and each $s\in\b(t)$, a decoding function is defined as
\begin{align}
g_t^{s}:\prod\limits_{e\in\In(t)}\{1,2,\ldots,2^{nC_e}\} \to\Xc_s.
\end{align}
\end{enumerate}
A rate vector $\mathbf{R}$ corresponding to the set $\{R_s\}_{s\in\Sc}$ is said to be achievable on network $\Nc$, if for any $\e>0$, there exists $n$ large enough and a coding scheme of block length $n$ such that for all $t\in\Tc$ and $s\in\b(t)$
\begin{align}
\P(\hat{M}_s^{(t)}\neq X_s )\leq \e,
\end{align}
where $\hat{M}_s^{(t)}$ denotes the reconstruction of message $M_s$ at node $t$. Let $\Rc(\Nc)$ denote the set of achievable rates on network $\Nc$.

Throughout the paper, vectors are denoted by bold upper-case letters, e.g.~$\mathbf{A}$, $\mathbf{B}$, etc. Sets are denotes by calligraphic upper-case letters, e.g. $\mathcal{A}$, $\mathcal{B}$, etc.  For a vector $\mathbf{A}=(a_1,a_2,\ldots,a_n)$ of length $n$ and a set $\mathbf{F}\subset\{1,2,\ldots,n\}$, $\mathbf{A}_{\mathcal{F}}$ denotes a vector of length $|\Fc|$ formed  by the elements of the vector $\mathbf{A}$ whose indices are in the set $\mathcal{F}$ in the order they show up in $\mathbf{A}$.
\section{Equivalent and bounding networks}\label{s:e-b-networks}

The problem we consider  is defined formally as follows.  For a given network $\Nc$, we wish to find a network $\Nc'$ with fewer links for which the set of achievable rates either bounds $\Rc(\Nc)$ from below ($\Rc(\Nc')\subseteq\Rc(\Nc)$), bounds $\Rc(\Nc)$ from above ($\Rc(\Nc)\subseteq\Rc(\Nc')$) or describes it perfectly ($\Rc(\Nc)=\Rc(\Nc')$). We take a hierarchical approach, sequentially applying  operations to simplify the given network. Following \cite{KoetterE:10}, each operation replaces a subnetwork of the network with a bounding model. Subnetwork $\Nc_2$ is an upper bounding model for a subnetwork $\Nc_1$ with the same number of input (source) and output (sink) nodes (written $\Nc_1\subseteq \Nc_2$) if all functions $\{f_t\}_{t\in\Tc}$ of sources that can be reconstructed at the sinks of $\Nc_1$ can also be reconstructed at the sinks of $\Nc_2$ can also be reconstructed at the sinks of $\Nc_2$.  Here, each function $f_t$, for $t\in\Tc$, is a function of the information sources that are available at source nodes. Subnetworks $\Nc_1$ and $\Nc_2$ are equivalent if $\Nc_1\subseteq\Nc_2$ and $\Nc_2\subseteq\Nc_1$.

When deriving upper and lower bounding networks, it is  desirable to find upper and lower bounding networks that have the same topologies. In this case, we can bound the difference between the capacity of a network $\Nc$ and capacities $\Rc(\Nc_l)$ and $\Rc(\Nc_u)$ of lower and upper bounding networks $\Nc_l$ and $\Nc_u$ using a bound from \cite{Effros:10}. Note that by having identical graphs, we also require that all links have non-zero capacity in both networks.

For comparing two networks $\Nc_l$ and $\Nc_u$ which have identical topologies, define the difference factor between $\Nc_l$ and $\Nc_u$ as
\begin{align}
\Delta(\Nc_l,\Nc_u)\triangleq\max\limits_{e\in\Ec} \frac{C_e(\Nc_u)}{C_e(\Nc_l)},
\end{align}
where $C_e(\Nc_l)$ and $C_e(\Nc_u)$ denote the capacities of the link $e$ is $\Nc_l$ and $\Nc_u$ respectively. Note that $\Delta(\Nc_l,\Nc_u)\geq1$. Let $\Rc_l$ ($\Rc_u$) denote the capacity region of $\Nc_l$ ($\Nc_u$). Then
\begin{align}
\Rc_l\subseteq\Rc\subseteq\Rc_u,
\end{align}
while
\begin{align}
\Rc_u\subseteq\Delta(\Nc_l,\Nc_u)\Rc_l.
\end{align}

\section{Basic simplification operations}\label{s:simple}

One of the simplest operations for deriving an upper-bounding  network for a given network is merging nodes. Coalescing two nodes is equivalent to adding two links of infinite capacity from each of them to the other one. This is precisely the approach employed in cut-set bounds. Because of these infinite-capacity links, combining nodes, unless done wisely, potentially can result in very loose bounds. However, we show that  in some cases, nodes can be combined without affecting the network capacity.  One simple example is when the sum of the capacities of the incoming links of a node $v$ is less that the capacity of each of its outgoing links. In this situation, the node can be combined with all nodes $w$ such that $(v,w)\in\Ec$.

Another possible operation for getting upper or lower bounding networks is reducing or increasing the link capacities. As a special case of such operations, one can reduce the capacity of a link to zero which is the same as deleting the link.  In some cases reducing/increasing link capacities is helpful in simplifying the network.

Another type of operation for simplifying networks  is based on network cut-sets. A cut $P$ between two sets of nodes $\Wc_1$ and $\Wc_2$ is a partition of the network nodes $\Vc$ into two sets $Vc_1$ and $\Vc_2$ such that $\Wc_1\subseteq \Vc_1,\Wc_2\subseteq \Vc_2$. The capacity of a cut is defined as the sum of capacities of the  forward links of the cut, i.e.~links $(v,w)$ such that $v\in\Vc_1,w\in\Vc_2$.  Links $(v,w)$ such that $v\in\Vc_2,w\in\Vc_1$ are called backward links of the cut. For example, if we find  a minimum cut separating a sink from its corresponding sources and all other sink nodes, and connect the forward links directly to the sink node while preserving all the backward links, this results in an upper-bounding network. If instead of keeping the backward links, we delete them, a lower-bounding network is obtained. In the case where there are no backward links, this procedure results in an equivalent network. We can of course repeat this procedure for every sink.

Another simplification operation involves removing a set $\Ac$ of links or components and possibly replacing it with additional capacity that might be spread over the remaining network. A simple lower bounding network can be obtained by removing a set  $\Ac$, while an
upper bounding network can be obtained from replacing a set $\Ac$ by adding sufficient capacity to the remaining network to be able to carry any
information that could have been carried by the set $\Ac$. For example, if the remaining network contains paths  from the inputs to the outputs
of the set $\Ac$, we can formulate a linear program, based on generalized flows, to find the minimum factor $k$  by which scaling up the
capacities of the remaining links uniformly gives an upper bounding network\footnote{One way to think of this is to associate a commodity with
each link or path segment in $\Ac$. Conversions between commodities take place according to the link capacity ratios in $\Ac$.  The linear
program minimizes $k$ subject to flow conservation of these commodities.}. Since the lower bounding network obtained by just removing the set
$\Ac$ differs from the upper bounding network by the scaling factor $k$, this gives a multiplicative bound of $k-1$ on the potential capacity
difference associated with the upper and lower bounding operations.

\section{Y-networks and generalizations}\label{s:Y}

Consider the network shown in Fig.~\ref{fig:3-node} consisting of four nodes and three directed links with capacities $(r_1,r_2,r_3)$. This topology is probably the simplest network in which nodes are sharing resources to send their information. We will refer to such a network as a Y-network $Y(r_1,r_2,r_3)$.

Now consider the network shown in Fig.~\ref{fig:5-node} which consists of two Y-networks with shared input and output nodes. The following lemma shows that in some special cases this network is equivalent to another Y-network. This simplification reduces the number of links by 3 and the number of nodes by 1.

\begin{lemma}\label{lemma:Ynet}
Consider the network shown in Fig. \ref{fig:5-node}, $\Nc_1$, when $\tilde{a}=\alpha a$, $\tilde{b}=\alpha b$ and $\tilde{c}=\alpha c$. This network is equivalent to $\Nc_2$, a Y-network  $Y((1+\alpha)a, (1+\alpha)b,(1+\alpha)c)$.
\end{lemma}

\begin{proof}
Clearly, a Y-network $Y((1+\alpha)a, (1+\alpha)b,(1+\alpha)c)$ is an outer bound for the network of Fig.~\ref{fig:5-node}. Hence, we only need to show that it also serves as  an inner bound.

For the rest of the proof assume that $\alpha$ is a rational number (If it is not rational, it can be approximated by rational numbers with arbitrary precision).

Consider a code of block length $n$ that runs on network $\Nc_2$. The middle node maps the $n(1+\a)a$ bits received from $x_1$ and the $n(1+\a)b$ bits received from $x_2$ to $n(1+\a)c$ bits  sent to node $y$.

In order to run the same code on $\Nc_1$, consider using the network $m$ times, where $m$ is chosen such that $\a m/(1+\a)n$ and $m/(1+\a)n$ are both integers. Note that since $\a$ is rational by assumption, it is always possible to find such $m$. Let $k_1\triangleq\a m/(1+\a)n$ and $k_2 \triangleq m/(1+\a)n$.

During these $m$ channel uses the intersection node at the left hand side receives $ma$ bits from $x_1$ and $mb$ bits from $x_2$.  This is equal to the bits received by the intersection node in $\Nc_2$ during $ma/(1+\a)an=k_2$ coding sessions. Therefore, using the code used on $\Nc_2$, these bits can be mapped into $k_2n(1+\a)c=mc$ bits that will be sent to $y$. Similarly, the number of bits received by the intersection node on the right hand side during $m$ channel uses is equal to the bits that would have been received by the intercession node on $\Nc_2$ during $k_1$ coding sessions.

\end{proof}


\begin{figure}[t]
\begin{center}
\psfrag{a}[l]{$r_1$}
\psfrag{b}[c]{$r_2$}
\psfrag{c}[c]{$r_3$}
\psfrag{x}[l]{$x_1$}
\psfrag{y}[l]{$x_2$}
\psfrag{z}[l]{$y$}
\includegraphics[width=2cm]{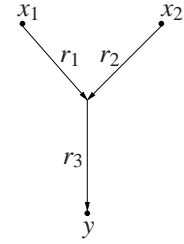}\caption{A Y-network}\label{fig:3-node}
\end{center}
\end{figure}

Lemma \ref{lemma:Ynet} serves as a useful tool in our network simplifications. For an example of how to employ this result, consider the network shown in Fig.~\ref{fig:4-node}. This network can be considered as a combination of two overlapping Y-networks.

\begin{figure}[h]
\begin{center}
\psfrag{a1}[l]{$a$}
\psfrag{a2}[l]{$\tilde{b}$}
\psfrag{b1}[c]{$b$}
\psfrag{b2}[l]{$\tilde{a}$}
\psfrag{c1}[l]{$c$}
\psfrag{c2}[l]{$\tilde{c}$}
\psfrag{x}[r]{$x_1$}
\psfrag{y}[r]{$x_2$}
\psfrag{z}[l]{$y$}
\includegraphics[width=3cm]{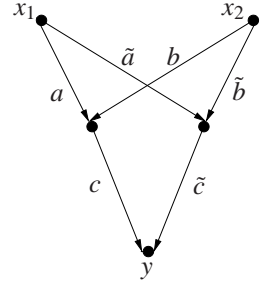}\caption{2 Y-networks with shared source and sink nodes and separate relay links}\label{fig:5-node}
\end{center}
\end{figure}

\begin{figure}[h]
\begin{center}
\psfrag{a}[l]{$a$}
\psfrag{b}[l]{$b$}
\psfrag{b1}[c]{$b'$}
\psfrag{c}[l]{$c$}
\psfrag{d}[l]{$d$}
\psfrag{x}[r]{$x_1$}
\psfrag{y}[r]{$x_2$}
\psfrag{z}[c]{$y$}
\includegraphics[width=2.3cm]{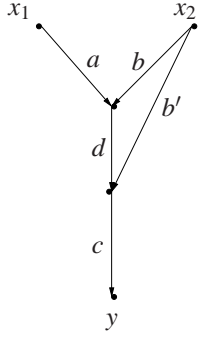}\caption{Two nodes communicating with one sink node via two relay nodes}\label{fig:4-node}
\end{center}
\end{figure}

\begin{lemma}\label{lemma:4node-eq}
Let $\b\triangleq\frac{b'}{b+b'}$. If $\b a + (1-\b)c \leq d$, then the network shown in Fig.~\ref{fig:4-node} is equivalent to a Y-network  $Y(a,b+b',c)$.
\end{lemma}
\begin{proof}
Clearly a Y-network $Y(a,b+b',c)$ is an upper-bounding network for the network of Fig.~\ref{fig:4-node}. We show that if the constraints in the lemma is satisfied,  it also serves as a lower-bounding network.

To find a lower bounding network, consider breaking  the links in Fig.~\ref{fig:4-node} into parallel links as in Fig.~\ref{fig:4-node-links-divided}, where $a=a_1+a_2+a_3$, $b=b_1+b_2+b_3$, etc. The network contains two capacity-disjoint Y-networks as illustrated in  Fig.~\ref{fig:4-node-2-Y}.  Our goal is to combine these two Y-networks by applying Lemma \ref{lemma:Ynet}, and in order to be able to do this, we require
$a_2=d_2=\a a_1$, $b'_1=\alpha b_1$ and $c_1=d_1={c_2\over\a}$. The combination of these two Y-networks will be a Y-network $Y((1+\a)a_1,(1+\a)b_1,(1+\a)c_1)$ which is a lower-bounding network for our original network. Now choosing  $b_1=b$, $b'_1=b'$ and $\a=\frac{b'}{b}$ from the link capacities constraints, we should have
\begin{align}
(1+\frac{b'}{b})a_1 & \leq a,\nonumber\\
(1+\frac{b'}{b})c_1 & \leq c,\nonumber\\
\frac{b'}{b}a_1+c_1 & \leq d.
\end{align}

From these inequalities, if $\b a + (1-\b)c \leq d$, we can choose $a_1=(1-\b)a$ and $c_1=(1-\b)c$ and get a lower-bounding Y-network $Y(a,b+b',c)$.

\begin{figure}[h]
\begin{center}
\psfrag{a1}[l]{\footnotesize{$a_1$}}
\psfrag{a2}[l]{\footnotesize{$a_2$}}
\psfrag{a3}[l]{\footnotesize{$a_3$}}
\psfrag{b1}[l]{\footnotesize{$b_1$}}
\psfrag{b2}[l]{\footnotesize{$b_2$}}
\psfrag{b11}[l]{\footnotesize{$b'_1$}}
\psfrag{b12}[l]{\footnotesize{$b'_2$}}
\psfrag{c1}[l]{\footnotesize{$c_1$}}
\psfrag{c2}[l]{\footnotesize{$c_2$}}
\psfrag{c3}[l]{\footnotesize{$c_3$}}
\psfrag{d1}[l]{\footnotesize{$d_1$}}
\psfrag{d2}[l]{\footnotesize{$d_2$}}
\psfrag{d3}[l]{\footnotesize{$d_3$}}
\psfrag{x}[r]{}
\psfrag{y}[r]{}
\psfrag{z}[c]{}
\includegraphics[width=3.5cm]{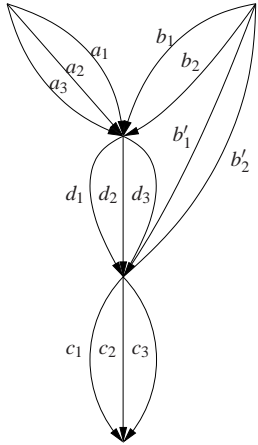}\caption{Breaking the links in Fig.~\ref{fig:4-node}  into parallel links }\label{fig:4-node-links-divided}
\end{center}
\end{figure}

\begin{figure}
\centering
\psfrag{a1}[l]{$a_1$}
\psfrag{a2}[l]{$a_2$}
\psfrag{a3}[l]{$a_3$}
\psfrag{b1}[l]{$b_1$}
\psfrag{b2}[l]{$b_2$}
\psfrag{b11}[l]{$b'_1$}
\psfrag{b12}[l]{$b'_2$}
\psfrag{c1}[l]{$c_1$}
\psfrag{c2}[l]{$c_2$}
\psfrag{c3}[l]{$c_3$}
\psfrag{d1}[l]{$d_1$}
\psfrag{d2}[l]{$d_2$}
\psfrag{d3}[l]{$d_3$}
\mbox{\subfigure{\includegraphics[width=2.5cm]{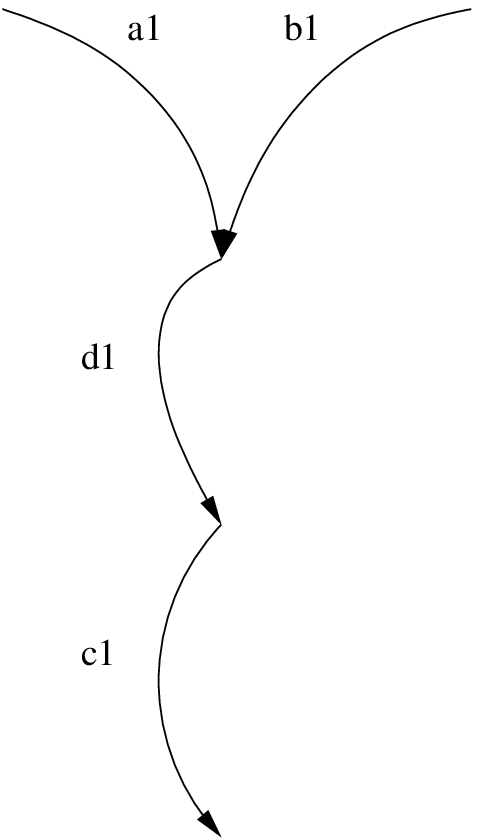}}\quad
\subfigure{\includegraphics[width=2.5cm]{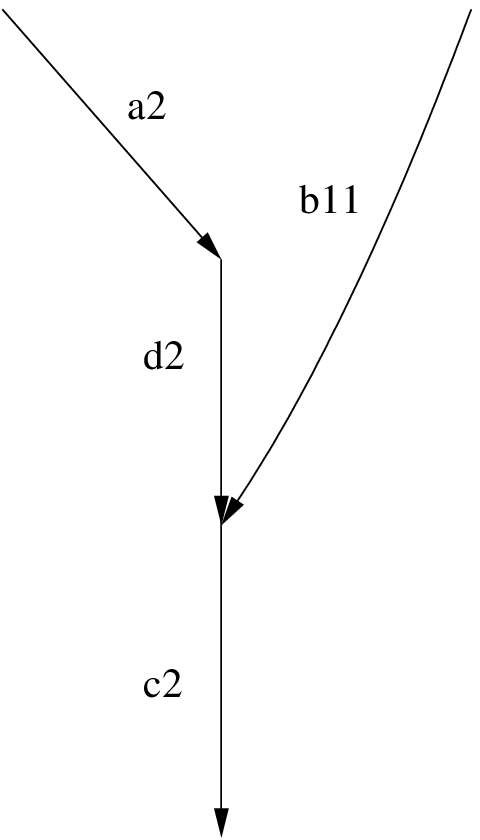} }}
\caption{Breaking the network in Fig.~\ref{fig:4-node} into two Y-networks} \label{fig:4-node-2-Y}
\end{figure}
\end{proof}

\begin{figure}[h]
\begin{center}
\psfrag{a}[l]{\footnotesize {$a$}}
\psfrag{b}[l]{\footnotesize {$b_0$}}
\psfrag{c}[l]{$c$}
\psfrag{d1}[l]{\footnotesize {$d_1$}}
\psfrag{d2}[l]{\footnotesize {$d_2$}}
\psfrag{dk}[l]{\footnotesize {$d_k$}}
\psfrag{e1}[l]{\footnotesize{$b_1$}}
\psfrag{e2}[l]{\footnotesize{$b_2$}}
\psfrag{ek}[l]{\footnotesize{$b_k$}}
\psfrag{o}[c]{$\vdots$}
\includegraphics[width=3.5cm]{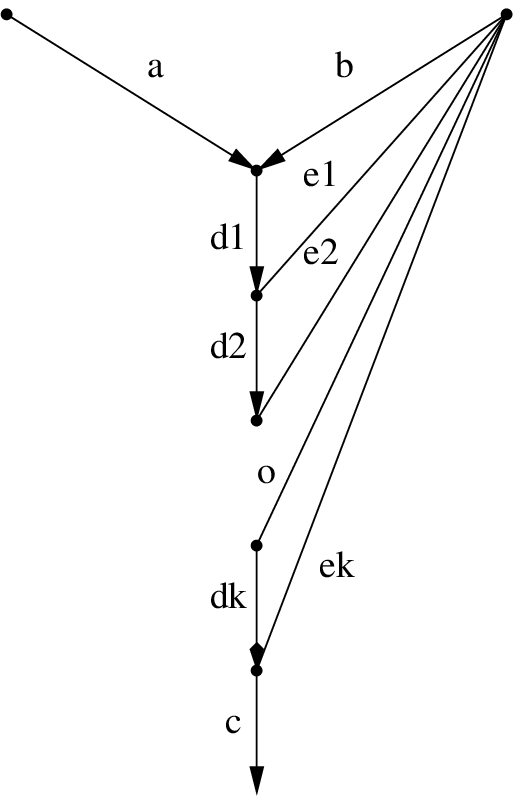}\caption{Generalization of the network of Fig.~\ref{fig:4-node}}\label{fig:4-node-gen}
\end{center}
\end{figure}

In order to get a better understanding of the required constraint stated in Lemma \ref{lemma:4node-eq}, consider the following special cases:
\begin{itemize}
\item[i.] $b\ll b' \to d \geq a$,
\item[ii.] $b=b'\to d\geq \frac{a+c}{2}$,
\item[iii.] $b\gg b' \to d\geq c$.
\end{itemize}

Again consider the network shown in Fig.~\ref{fig:4-node} where all the links have capacity $1$ except for the link of capacity $d$. From Lemma \ref{lemma:4node-eq}, for $d \geq  0.5\times 1 +0.5\times 1=1$,  this network is equivalent to a Y-network  $T(1,2,1)$. Our proof approach in Lemma \ref{lemma:4node-eq} does not say anything about the case where $d<1$. It might be the case that even for some values $d<1$,  still this equivalence holds. As we show through an example this cannot happen, and for $d<1$ the Y-network of $Y(1,2,1)$ is a strict upper bound for our network. Assume that $b_1$  and $(b_2,b_3)$ are available at $x_1$ and $x_2$ respectively, where $b_i\in\{0,1\}$. The goal is reconstructing $b_1b_2+\bar{b}_1\bar{b}_2+b_3$ at node $y$, where all operations are in $GF(2)$. This can be done easily in the Y-network  $Y(1,2,1)$, but is impossible in the original network for $d<1$.


Fig.~\ref{fig:4-node-gen} shows a generalization of the network of Fig.~\ref{fig:4-node}, where instead of 2 intermediate nodes, there are $k+1$ intermediate nodes.  Let $\a_0\triangleq1$, $\a_i\triangleq b_i/b_0$, for $i\in\{1,\ldots,k\}$,
\[
a_1\triangleq\frac{a}{1+\sum\limits_{i=1}^k \alpha_i},
\]
and
\[
c_1\triangleq\frac{c}{1+\sum\limits_{i=1}^k \alpha_i}.
\]
By extending the proof of Lemma \ref{lemma:4node-eq} to this more general case, we get Lemma \ref{lemma:3}.
\begin{lemma}\label{lemma:3}
If, for $i\in\{1,\ldots,k\}$,
\begin{align}
\sum\limits_{j=0}^{i-1} \a_j c_1 + \sum\limits_{j=i}^{k} \a_j  a_1\leq d_i,
\end{align}
then the network shown in Fig.~\ref{fig:4-node-gen} is equivalent to a Y-network $Y(a,\sum\limits_{i=0}^k b_i,c)$.
\end{lemma}

Another possible generalization of a Y-network is shown in Fig.~\ref{fig:4-node-gen2}. Here, while the number of relay nodes is kept as two, the number of source nodes is increased.  For this network, we can prove Lemma \ref{lemma:4} with straightforward extension of Lemma \ref{lemma:4node-eq}.
\begin{lemma}\label{lemma:4}
For $i=\{2,\ldots,k\}$, let
\[
\beta_i \triangleq \frac{b_i}{b'_i}.
\]
Then, if
\begin{align}
d \geq \sum\limits_{i=1}^n \b_i a+ \sum\limits_{i=1}^n (1-\b_i) c,
\end{align}
the two intermediate nodes can be combined without changing the performance.
\end{lemma}

\begin{figure}[h]
\begin{center}
\psfrag{x}[r]{\footnotesize {$x$}}
\psfrag{y}[r]{\footnotesize {$y$}}
\psfrag{z}[r]{\footnotesize {$z$}}
\psfrag{a}[l]{\footnotesize {$a$}}
\psfrag{b}[l]{\footnotesize {$b_1$}}
\psfrag{b1}[l]{\footnotesize {$b'_1$}}
\psfrag{b2}[l]{\footnotesize {$b_2$}}
\psfrag{b3}[l]{\footnotesize {$b'_2$}}
\psfrag{b4}[r]{\footnotesize {$b_k$}}
\psfrag{b5}[r]{\footnotesize {$b'_k$}}
\psfrag{dot}[c]{\footnotesize{$\ldots$}}
\psfrag{o}[c]{$\vdots$}
\includegraphics[width=5.5cm]{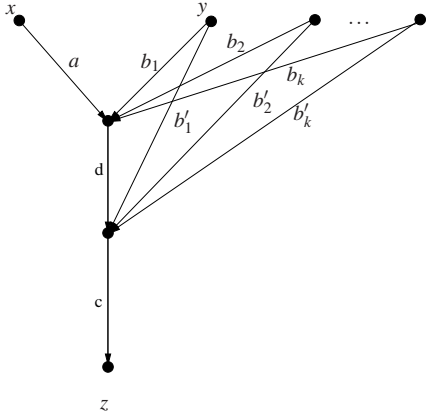}\caption{Another generalization of the network of Fig.~\ref{fig:4-node}}\label{fig:4-node-gen2}
\end{center}
\end{figure}

Again consider the network shown in Fig.~\ref{fig:4-node}, and assume that $d<\b a + (1-\b)c$. As shown in the previous section, a Y-network $Y(a,b+b',c)$, $\Nc_u$, serves as a strict upper bound for the network of Fig.~\ref{fig:4-node}.  In order to get a lower bounding Y-network,  we reduce the  capacities of the links $a$ and $c$ by a factor $\delta$, such that $d=\b \d a + (1-\b)\d c$. Hence,
\[
\d = \frac{d}{\b a + (1-\b)c},
\]
and from our assumption $0 < \d < 1 $. Using this, we can again invoke Lemma \ref{lemma:4node-eq} and get a lower-bounding Y-network  of $(\d a, b+b',\d c)$ denoted by $\Nc_l$. Comparing the link capacities of the upper and lower bounding Y-networks, their difference factor is $\Delta(\Nc_l,\Nc_u)=\d$. If $d\ll \b a + (1-\b)c$, then $\d \ll 1$,  meaning that the bounds become very loose in such cases.

To solve this problem, consider another possible pair of upper and lower bounding networks shown in Fig.~\ref{fig:4-node-bounds}, which have a topology different than a Y-network. Here the assumption is that $d<c$. It is easy to check that these are indeed upper and lower bounding networks. For deriving the upper-bound, the incoming information from links $a$ and $b$ are assumed to be transmitted directly to the final node, and since the incoming capacity of the final node is $c$, the information sent to it from link $b'$ can be captured by a direct link of capacity $\min(b',c)$.   For the lower bound, since $d<c$, all the information on link $c$ can be directed to the final node. By this strategy, the remaining unused capacity of link $c$ is $c-d$ which can be dedicated to link $b'$. If $\min(c-d,b')=b'$, then $c\leq d+b'$, then the upper and lower bounding networks coincide, and we have an equivalent network. The more interesting case is when $c > d+b'$, and therefore $\min(c-d,b')=c-d$.  In this case the difference factor of the upper and lower bounding networks is at least $1-d/c$. Choosing the best bound depending on the link capacities, given $a,b,b'$ and $c$ the worst difference factor is
\begin{align}
\min\limits_d \max\{ \frac{d}{\b a + (1-\b)c},1-\frac{d}{c}\} =\frac{c}{c+\beta a + (1-\b)c}.
\end{align}
As an example, for the case where $a=b=b'=c=1$, the worst case difference factor is $0.5$ which corresponds to $d=0.5$. This means that choosing the best pair of bounds for different values of $d$, the difference factor of our selected pair is always lower bounded by $0.5$.

\begin{figure}
\centering
\psfrag{a}[l]{$a$}
\psfrag{b}[c]{$b$}
\psfrag{c}[c]{$d$}
\psfrag{e1}[l]{$\min(b',c)$}
\psfrag{e2}[l]{$\min(b',c-d)$}
\mbox{\subfigure{\includegraphics[width=2.5cm]{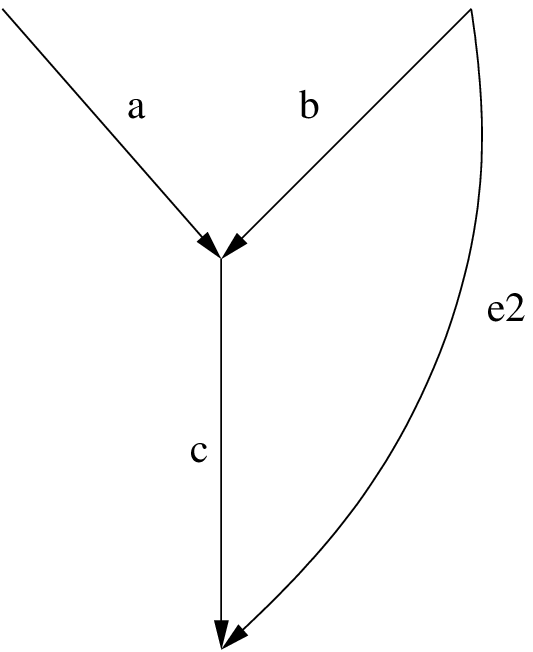}}\quad\quad\quad
\subfigure{\includegraphics[width=2.5cm]{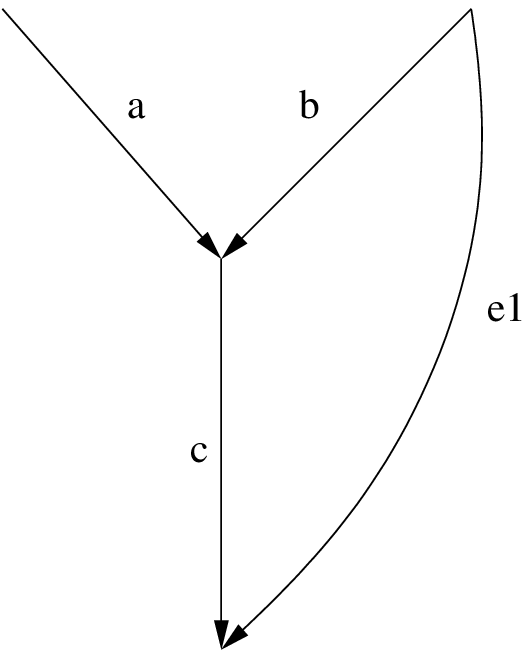} }}
\caption{A pair of upper and lower bounding networks for the network of Fig.~\ref{fig:4-node}.} \label{fig:4-node-bounds}
\end{figure}

\begin{figure}
\centering
  \subfigure[Network $\Nc_1$]{  \includegraphics[scale=0.35]{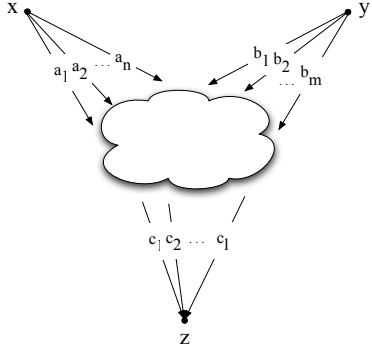}}\hspace{0.5 in}
   \subfigure[Example 1]{\includegraphics[scale=0.35]{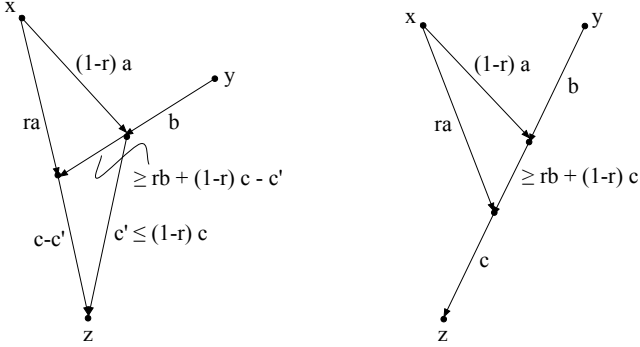}}\hspace{0.5 in}
  \subfigure[Example 2]{\includegraphics[scale=0.35]{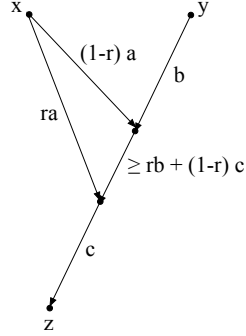}}\\
  \caption{A family of components that can be replaced with an equivalent component. The labels represent the link capacities, which satisfy $a =a_1+a_2+\dots a_n$, $b=b_1 +b_2+\dots b_m$, $c=c_1+c_2+\dots c_l$, and $0 <r <1$. }\label{figure:cloudy-y}\end{figure}

We conclude this section by discussing how the approach in the section generalizes to a larger class of network topologies.  Consider the network  $\Nc_1$ shown in  Fig.~\ref{figure:cloudy-y}. A Y-network $Y(\sum a_i,\sum b_i,\sum c_i)$ is always an upper bound for the network $\Nc_1$.  On the other hand, if we can show, as in the proof of Lemma~\ref{lemma:4node-eq}, that $\Nc_1$ contains two capacity-disjoint Y-networks of the form $Y(\alpha \sum a_i,\alpha \sum b_i,\alpha \sum c_i)$ and $Y((1 -\alpha) \sum a_i,(1 -\alpha) \sum b_i,(1 -\alpha) \sum c_i)$, then $Y(\sum a_i,\sum b_i,\sum c_i)$ is also a lower bound for, and thus equivalent to, $\Nc_1$.  Subfigures (b) and (c) provide two simple examples of networks where this is the case. This construction can also be generalized by replacing the basic $Y$-shaped topology with star-shaped topologies with arbitrary numbers of inputs and outputs.

\section{Effect of link capacities}\label{s:effect}

Consider two networks $\Nc_1$ and $\Nc_2$ which have identical topologies and link capacities except for some link $e$ which has capacity $C_{e,1}$  in $\Nc_1$ and capacity $C_{e,2}<C_{e,1}$ in $\Nc_2$. Let $\Rc_1$ and $\Rc_2$ denote the set of achievable rates on $\Nc_1$ and $\Nc_2$ respectively. The question is how this difference affects the performances of these two networks.  One way of doing this comparison is based on what was mentioned earlier, i.e., to compute the ratio between  $C_{e,1}$ and  $C_{e,2}$. However, this might not always give the best possible bound. The reason is that it might be case that while  $C_{e,1}$ and  $C_{e,2}$ are both small compared to the capacity region of the networks,  $\Delta(\Nc_l,\Nc_u)=C_{e,1}/C_{e,2}$ is very large. In this section, we study this problem in more details.

Note that the link of capacity $C_{e,1}$ in network $\Nc_1$ into two parallel links of capacities $C_{e,1}-C_{e,2}$ and $C_{e,2}$. Clearly this process does not affect the capacity region of $\Nc_1$. By this transformation, network $\Nc_2$ is equivalent to this new network with link of capacity $\e \triangleq C_{e,1}-C_{e,1}$ being removed. Therefore, in the rest of this section, instead of changing the capacity of a link, we assume that a link of capacity $\e$ is removed as shown in Fig.~\ref{fig:networks}, and prove that at least in some cases changing the capacity of a link by $\e$ cannot have an effect larger than $\e$ on the set of achievable rates, \ie if the rate vector $\mathbf{R}$ is achievable on $\Nc_1$, rate vector $\mathbf{R}-\e\mathbf{1}$, where $\mathbf{1}$ denotes an all-one vector of length $|\Sc|$, is achievable on $\Nc_2$ as well.

One such example is the case of multicast networks. In that case, the capacity of the network is determined by the values of the cuts from the sources to the sinks. Therefore, since removing a link of capacity $\e$ does not change the values of the cuts by more than $\e$, the capacity of the network will not be affected by more than $\e$.

Another example, is the case where all sources are connected directly by a link to a super-source node which has therefore access to the information of all sources.  Without loss of generality, let $\Sc=\{1,2,\ldots,|\Sc|\}$. As before, assume that each sink node $t\in\Tc$ is interested in recovering a subset of sources denoted by $\b(t)$.

\begin{theorem}
For the described network with a super source-node and arbitrary demands, removing a link of capacity $\e$ can change the capacity region by at most $\e$.
\end{theorem}

\begin{figure}[t]
\begin{center}
\psfrag{N1}[l]{$\mathcal{N}_1$}
\psfrag{N2}[l]{$\mathcal{N}_2$}
\psfrag{e}[l]{$e$}
\psfrag{v1}[l]{$v_1$}
\psfrag{v2}[l]{$v_2$}
\includegraphics[width=9cm]{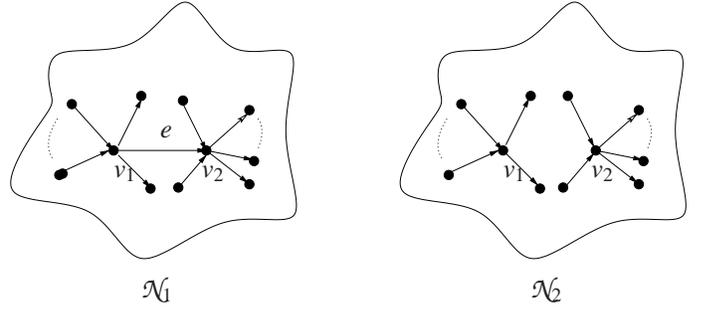}\caption{Networks $\Nc_1$ and $\Nc_2$ which are with and without link $e$ respectively.}\label{fig:networks}
\end{center}
\end{figure}

\begin{proof}

Assume that the rate vector $\mathbf{R}=(R_1,R_2,\ldots,R_{|\Sc|})$ is achievable on $\Nc_1$. Since in the single source case, capacity regions corresponding to zero and asymptotically zero probability of errors coincide \cite{ChanG:10}, we can assume  there exist a coding scheme of block length $n$ that achieves rate $\mathbf{R}$ on $\Nc_1$ with zero probability of error.

Based on this zero-error code, we construct another code for network $\Nc_2$ that achieves rate $\mathbf{R}-\e\mathbf{1}$ with asymptotically zero probability of error.

By our assumption about the network structure, the super source-node observes a  message vector $\mathbf{M}=(M_1,\ldots,M_{|\Sc|})$, where $M_s\in\{1,2,\ldots,2^{nR_s}\}$, for $s\in\{1,2,\ldots,|\Sc|\}$.

Now consider the link $e=(v_1,v_2)$ in network $\Nc_1$ which has been removed in network $\Nc_2$. During the coding process on $\Nc_1$, the bits sent across this link can take on at most $2^{n\e}$ different values. Consider binning the message vectors $\mathbf{M}=(M_1,M_2,\ldots,M_{|\Sc|})$ into $2^{n\e}$ different bins based on the bit stream  sent over this link during their transmission. Since the code is a deterministic code, each message vector only corresponds  to one bin. Since there exist $2^{n\sum_{s\in\Sc} R_s}$ distinct message vectors,  based on the Pigeonhole principle, there will be at least one bin with more that $2^{n(\sum_{s\in\Sc}R_s-\e)}$ message vectors. Denote  the message vectors contained in this bin by the set $\mathcal{M}_0$.

In $\Nc_2$, no message can be sent from $v_1$ to $v_2$ through link $e$. Therefore, in order to run the code for $\Nc_1$ on network $\Nc_2$, we need to specify the assumption of node $v_2$ about the message that it would have received from node $v_1$ in $\Nc_1$. Let node $v_2$ assume that it always receives the bit pattern corresponding to the bin containing $\mathcal{M}_0$. Making this assumption and having the rest of $\Nc_2$ to perform as in $\Nc_1$, it is clear that all message vectors in $\mathcal{M}_0$ can be delivered with zero probability of error on $\mathcal{N}_2$ as well. In other words, if the input to the super source node in $\Nc_2$ is one of the message vectors in $\mathcal{M}_0$, then each sink $t\in\Tc$ recovers its intended message $\mathbf{M}_{\b(t)}$ with probability one. In the rest of the proof we show how this set $\mathcal{M}_0$ can be used to deliver independent information to different receivers over network $\Nc_2$.

Define random vector $\mathbf{U}=(U_1,U_2,\ldots,U_{|\Sc|})$ to have a uniform distribution over the elements of $\mathcal{M}_0$. For each input vector $\mathbf{U}$, each sink node $t$, recovers $\mathbf{U}_{\b(t)}$ perfectly.

The described model with input $\mathbf{U}$, and ouputs $\mathbf{U}_{\b(t)}$, for $t\in\Tc$ is  a deterministic broadcast channel (DBC)\footnote{DBCs and their capacity regions are briefly described in Appendix A.} whose capacity region is known. Therefore, we can employ this DBC to deliver information on $\Nc_2$.

Before doing this, we slightly change the set of sink nodes, and replace the set of sinks $\Tc$ by $\Tc_e$ as described in the following. This modification does not affect the functionality of the network, but simplifies the statement of the proof. Divide  each sink node $t$ into $|\b(t)|$ sink nodes, such that each one has access to all the information available to the node $t$, but is only interested in reconstructing one of the sources. Let $\Tc_e$ denote this expanded set of sinks. Consider a subset $\Tc_s$of size $|\Sc|$ of $\Tc_e$ such that each source $s\in\Sc$ is recovered by one of the sinks in $\Tc_S$. Since each sink in $\Tc_s$ only recovers one source, hence there should be a one-to-one correspondence between the elements of $\Sc$ and $\Tc_s$.

Now consider the DBC  with input $\mathbf{U}$ and outputs $\{U_t\}_{t\in\Tc_s}$. Since the code is zero-error, and there is a one-to-one correspondence between the elements of $\Sc$ and $\Tc_s$,  $\{U_s\}_{s\in\Sc}$ can be replaced by $\{U_s\}_{s\in\Sc}$. The capacity region of this DBC, as explained in Appendix A, can be described by the set of rates $(r_1,r_2,\ldots,r_{|\Sc|})$ satisfying
\begin{align}
\sum_{s\in{\Ac}} r_s\leq H(U_{\Ac}),\label{eq:inequal}
\end{align}
for all $\Ac\subset\Sc$. From our construction,
\begin{align}
H(U_{\Sc})&=H(U_1,U_2,\ldots,U_{|\Sc|}) \nonumber\\
&n(\sum_{s\in\Sc}R_s-\e).
\end{align}
On the other hand, for each $\Ac\subset\Sc$
\begin{align}
H(U_{\Sc})&=H(U_{\Ac})+H(U_{\Sc\backslash\Ac}).
\end{align}
Therefore,
\begin{align}
H(U_{\Ac})&=H(U_{\Sc})-H(U_{\Sc\backslash\Ac})\nonumber\\
&\geq n(\sum_{s\in\Sc}R_s-\e) -\sum_{s\in\Sc\backslash\Ac}H(U_s)\nonumber\\
&\geq n(\sum_{s\in\Sc}R_s-\e) -\sum_{s\in\Sc\backslash\Ac}nR_s\nonumber\\
&= n(\sum_{s\in\Ac}R_s-\e).
\end{align}

It is easy to check that the point $(r_1,r_2,\ldots,r_{|\Sc|})=n(\mathbf{R}-\e\mathbf{1})$ satisfies all inequalities required by \eqref{eq:inequal}. Hence, using the network $\Nc_2$ as a DBC, the rate vector $\mathbf{R}-\e$ can be transmitted to nodes in $\Tc_s$.

We now argue that the described DBC code, with no extra effort, delivers the rate vector $\mathbf{R}-\e$ on network $\Nc_2$ to all sinks, not just those in $\Tc_s$. The reason is that for each message vector $\mathbf{M}\in\mathcal{M}_0$, and consequently for each input $\mathbf{U}$, all sink nodes in $\Tc_e$ that are interested  in recovering a source $s$, \ie for all $t\in\Tc_e$ such that $\b(t)= s$, receive $M_s$ with probability one. Hence, if by the described coding scheme  rate $R_s$ is delivered to one of them, then the rest are able to receive the same data stream at rate $R_s$ as well.
\end{proof}

For some simple cases of multi-source multi-demand case networks, for instance when the link of capacity $\e$ is directly connected to one of the sources, the same result still holds, i.e.~removing that link cannot change the capacity region by more than $\e$. However, for the general case, it remains open as to how much removing a link of capacity $\e$ can affect the capacity region of the network.

\section{Cycles and delays}\label{s:cycle}

In \cite{cover}, the capacity of a channel is defined in terms of bits per channel use. In the case where there are no cycles, this definition is equivalent to defining capacity in terms of bits transmitted per unit time. However, we do not know whether in a general network with cycles and delays, these two definitions, \ie bits per channel use versus bits per unit time, could result in different capacity regions. For example, consider two networks $\Nc_1$ and $\Nc_2$ that are identical except  for a single component cycle, illustrated in Fig.~\ref{fig:loops}.  In $\Nc_2$ the transmission delay from $v_1$ to $v_2$ is larger than in $\Nc_1$ due to the added node $v_0$, assuming that each link introduces unit delay. Now, we do not know whether achieving capacity on $\Nc_1$ could require  interactive communication over a feedback loop, where each successive symbol transmitted by $v_1$ depends on previous symbols received from $v_2$ and vice versa. If we try to run the same code  on $\Nc_2$,  the number of symbols per unit time that can be exchanged by nodes $v_1$ and $v_2$  is decreased because of  the additional delay introduced by node $v_0$. Hence, if the rate is defined in terms of bits communicated per unit time, this additional delay will reduce the capacity region. However, if the capacity is defined in terms of bits communicated per channel use, by letting each node wait until it has received the information it requires to generate its next symbol, the delivered  rates and consequently the capacity region will not be affected. \footnote{Note that in the case of multicast, since it is  known that capacity-achieving codes do not require such communication over feedback loops, the capacity region remains unchanged under both definitions, even in the case of cyclic graphs.}

In most of the operations introduced in this paper,  the number of hops between nodes is different in the original network compared to its simplified version.  Therefore, such operations in a cyclic graph with delay can potentially change the capacity region in ways that are not predicted by our analysis. Hence, in this paper we have restricted ourselves to the case of acyclic graphs. However, studying the effect of these operations in a general network with cycles remains as a question for further study.

\begin{figure}
\centering
  \subfigure[Component of $\Nc_1$]{\psfrag{v1}[l]{$v_1$}\psfrag{v2}[l]{$v_2$}
\includegraphics[width=3.5cm]{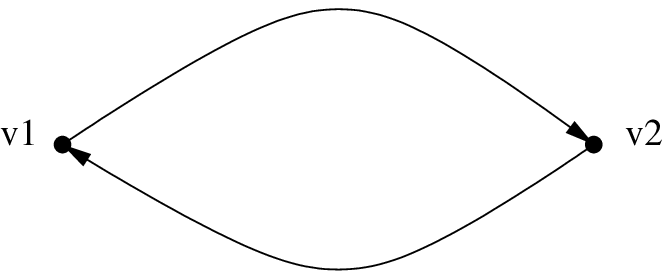}}\hspace{0.2cm}
  \subfigure[Component of $\Nc_2$]{\psfrag{v1}[l]{$v_1$}\psfrag{v2}[l]{$v_2$}\psfrag{v3}[l]{$v_0$}\includegraphics[width=3.5cm]{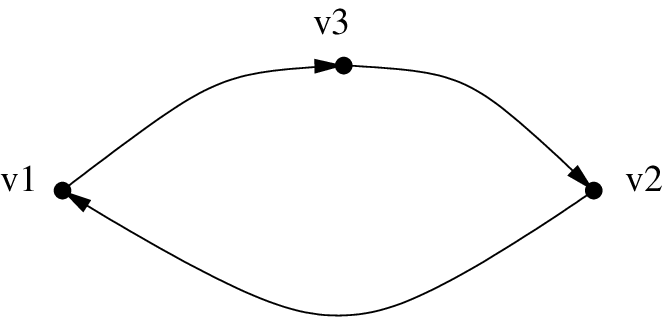}}
   \caption{Changing the delay between two nodes by introducing an extra node}\label{fig:loops}
   \end{figure}

\section{Conclusions}\label{s:conclude}

In this paper,  we proposed new techniques for efficiently approximating or bounding the capacity of networks with complex topologies
and arbitrary communication demands, such as non-multicast and functional demands.  We proposed to take a new approach based on
systematic graph operations that can be applied recursively to simplify large networks, so as to facilitate the application of computational
bounding tools. Besides its generality, another motivation for such a recursive graph-based approach is  that it lends itself to tractable
``divide-and-conquer'' algorithms for analyzing very large networks and allows  computation to be traded off against tightness of obtained
bounds according to the extent of simplification.

Techniques proposed for networks of directed noiseless point-to-point links (bit pipes) can be readily applied to networks of noisy multi-terminal channels by using the results of Koetter et al.~\cite{KoetterE:10} to replace each noisy link with a bounding model consisting of such bit pipes.  However, it may be possible to obtain better bounds by developing simplification tools that are directly applicable to noisy networks.  Thus, we will seek to extend the above work to networks containing noisy broadcast, multiple access and interference links.

\renewcommand{\theequation}{A-\arabic{equation}}
\setcounter{equation}{0}  

\section*{APPENDIX A\\ Deterministic broadcast channel}  \label{app1} 

Deterministic broadcast channels (DBC) are spacial cases of general broadcast channels. In a $K$-user DBC,
\[
\P((Y_1,\ldots,Y_K)=(y_1,\ldots,y_K)|X=x)\in\{0,1\}.
\]
In other words, since in a BC the capacity region only depends on the marginal distributions \cite{cover}, a $K$-user DBC can be described by $k$ functions $(f_1,\ldots,f_k)$ where
\begin{align}
f_i:\Xc\to\Yc_i,\label{eq:map}
\end{align}
and $Y_i=f_i(X)$ for $i\in\{1,2,\ldots,k\}$. In \eqref{eq:map},  $\Xc$ and $\Yc_i$ refer to the channel input alphabet and the output alphabet of the $i^{\rm th}$ channel respectively.

The capacity region of a $k$-user DBC can be described by the union of the set of rates $(R_1,R_2,\ldots,R_k)$ satisfying
\begin{align}
\sum\limits_{i\in\mathcal{A} }R_i \leq H(Y_{\mathcal{A}}), \forall \mathcal{A}\subset\{1,\ldots,k\},
\end{align}
for some $P(X)$  \cite{Marton:77,Pinsker:78} .

\begin{figure}
\begin{center}
\psfrag{X}[l]{\scriptsize{$X$}}
\psfrag{Y1}[l]{\scriptsize{$Y_1=f_1(X)$}}
\psfrag{Y2}[l]{\scriptsize{$Y_2=f_2(X)$}}
\psfrag{Y3}[l]{\scriptsize{$Y_{K-1}=f_{K-1}(X)$}}
\psfrag{Y4}[l]{\scriptsize{$Y_{K}=f_K(X)$}}
\psfrag{BC}[l]{DBC}
\includegraphics[width=4.5cm]{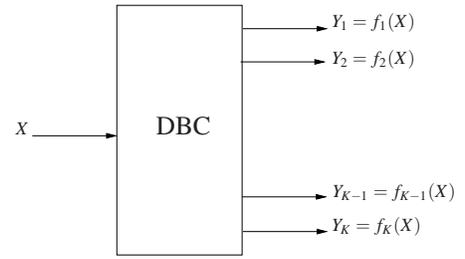}\caption{Deterministic broadcast channel}\label{fig:BC}
\end{center}
\end{figure}

\section*{Acknowledgments}
This work has been supported in part by the Air Force Office of Scientific Research under grant FA9550-10-1-0166 and Caltech's Lee Center for Advanced Networking.

\bibliographystyle{unsrt}
\bibliography{myrefs}

\end{document}